\def\complexNumbers{\mathbb{C}}

\def\constante{{\rm e}}
\def\constantj{{\rm j}}

\def\positiveInt{m}
\def\psksymbol[#1]{d_{#1}}
\def\psksize{H}
\def\coefficienta[#1]{{{d}}_{#1}}

\def\chirpm{m}
\def\chirpn{n}
\def\chirpphase[#1][#2]{{\psi_{#1}{(#2)}}}
\def\seqx[#1]{\textit{{x}}(#1)}
\def\seqy[#1]{\textit{{y}}(#1)}

\def\symbolPSK[#1]{s_{#1}}
\def\chirpndetect{\hat{n}}
\def\chirpmdetect{\hat{m}}
\def\symbolPSKdetect[#1]{\hat{s}_{#1}}
\def\dataSymbolAfterIDFTspread[#1]{\tilde{d}_{#1}}
\def\test[#1]{C_{#1}}
\def\numberofIndices{l}
\def\spectralEfficiency{\rho}
\def\EbNO{E_{\text{b}}/N_{\text{0}}}

\def\transmittedSignal[#1]{p\left(#1\right)}
\def\numberOfOccupiedSubcarriers{D}

\def\numberOfShifts{M}
\def\indexSubcarrier{k}
\def\indexpsksymbol{k}
\def\indexTime{m}
\def\basisFunction[#1]{B_{#1}(\timeVar)}
\def\amountOfShift[#1]{\tau_{#1}}
\def\dataSymbols[#1]{d_{#1}}
\def\angleSignal[#1]{\psi_{#1}(\timeVar)}
\def\instantaneousFrequency[#1]{F_{#1}(t)}
\def\besselFunctionFirstKind[#1][#2]{J_{#1}\left(#2\right)}
\def\lowerFrequency{L_{\rm d}}
\def\upperFrequency{L_{\rm u}}
\def\idftSize{N}
\def\indexSample{n}
\def\fourierSeries[#1]{c_{#1}}

\def\symbolDuration{T_{\rm s}}
\def\timeVar{t}

\def\fresnelC[#1]{C(#1)}
\def\fresnelS[#1]{S(#1)}
\def\linearXone{\alpha_\indexSubcarrier}
\def\linearXtwo{\beta_\indexSubcarrier}
\def\linearCoef{\gamma_\indexSubcarrier}

\def\fcarrier{f_{\rm c}}

\def\complexNumbers{\mathbb{C}}

\def\integersPositiveSet{\mathbb{Z}^{+}}

\def\constante{{\rm e}}
\def\constantj{{\rm j}}

\def\psksize{H}

\def\lagForCorrelation{k}

\def\indexEleOfSeq{i}

\def\indexIterationANF{l}

\def\orderMonomial[#1]{k_{#1}}
\def\coeffientsANF[#1]{c_{#1}}
\def\polyVariable{z}

\def\numberOfPointsForPSK{H}
\def\modulationSymbolF[#1]{m_{#1}}
\def\cardinalitySetOfOperators[#1]{{H}_{#1}}

\def\monomial[#1]{x_{#1}}

\def\lengthGaGb{M}

\def\scaleAexp[#1]{a_{#1}}
\def\scaleBexp[#1]{b_{#1}}
\def\scaleEexp[#1]{e_{#1}}
\def\angleexp[#1]{c_{#1}}
\def\angleexpAll[#1]{k_{#1}}
\def\angleScaleAexp[#1]{\dot{c}_{#1}}
\def\angleScaleBexp[#1]{\ddot{c}_{#1}}

\def\separationGolay[#1]{d_{#1}}

\def\scaleEexpBatch[#1][#2]{e_{#1}^{(#2)}}
\def\angleexpAllBatch[#1][#2]{k_{#1}^{(#2)}}

\def\scaleEexpPre[#1]{\dot{e}_{#1}}
\def\angleexpAllPre[#1]{\dot{k}_{#1}}

\def\eleGa[#1]{{a}_{#1}}
\def\eleGb[#1]{{b}_{#1}}

\def\OCB{{M}_{\rm ocb}}

\def\apac[#1][#2]{\rho_{#1}(#2)}
\def\apacPositive[#1][#2]{\rho^{+}_{#1}(#2)}
\def\binaryAsignment[#1][#2]{b_{#1}^{(#2)}}

\def\eleSeqf[#1]{{f}_{#1}}
\def\eleSeqg[#1]{{g}_{#1}}
\def\eleSeqcf[#1]{{c}_{f,#1}}
\def\eleSeqcg[#1]{{c}_{g,#1}}

\def\scaleA[#1]{\alpha_{#1}}
\def\scaleB[#1]{\beta_{#1}}
\def\angleGolay[#1]{\omega_{#1}}
\def\angleScaleA[#1][#2]{\dot{\omega}_{#1}^{#2}}
\def\angleScaleB[#1][#2]{\ddot{\omega}_{#1}^{#2}}
\def\permutationShift[#1]{{\psi_{#1}}}
\def\permutationMono[#1]{{p_{#1}}}

\def\CPDuration{T_{\rm cp}}
\def\timeVar{t}

\def\funczArbitrary[#1]{z(#1)}
\def\funcaArbitrary[#1]{a(#1)}
\def\funcbArbitrary[#1]{b(#1)}
\def\coefficientArbitrary[#1]{k_{#1}}

\def\distanceToPoint[#1]{d_{#1}}
\def\ratioBetweenDistanceAndInner[#1]{l_{#1}}
\def\angleBetweenPointAndXaxis[#1]{\theta_{#1}}
\def\angleBetweenPointAndXYDiagonal[#1]{\psi_{#1}}
\def\angleBetweenPointAndSymPoint[#1]{\xi_{#1}}

\def\setPSKsymbols[#1]{\mathbb{S}_{{\rm PSK},#1}}
\def\setChirps{\mathbb{S}_{\rm chirp}}

\def\numberOfInfoBits{S}
\def\averagePower{P_{\rm av}}

\def\vecArrangement[#1]{\textbf{b}_{#1}}

\def\seqGa{\textit{\textbf{a}}}
\def\seqGb{\textit{\textbf{b}}}
\def\seqGaIt[#1]{\textit{\textbf{a}}^{(#1)}}
\def\seqGbIt[#1]{\textit{\textbf{b}}^{(#1)}}

\def\seqGf[#1]{\textit{\textbf{f}}_{#1}}
\def\seqGg[#1]{\textit{\textbf{g}}_{#1}}
\def\seqGfdot[#1]{\bar{\textit{\textbf{f}}}_{#1}}
\def\seqGgdot[#1]{\bar{\textit{\textbf{g}}}_{#1}}

\def\seqSub[#1]{\textit{\textbf{h}}_{#1}}

\def\seqFirstOrderMonomial[#1]{\textit{\textbf{m}}_{#1}}

\def\seqToBeModulated[#1]{\textit{\textbf{s}}_{#1}}

\def\fourier[#1]{\mathcal{F}\{#1\}}
\def\flipConjugate[#1]{{{\tilde{#1}}}}
\def\expectationOperator[#1]{{\mathbb{E}}[#1]}
\def\operator[#1][#2]{\mathcal{O}_{#1}^{(#2)}}
\def\operatordot[#1][#2]{\bar{\mathcal{O}}_{#1}^{(#2)}}

\def\compositeOperatorF[#1][#2]{{F}_{#1}{(#2)}}
\def\compositeOperatorG[#1][#2]{{G}_{#1}{(#2)}}
\def\compositeOperatorFdot[#1][#2]{\bar{F}_{#1}{(#2)}}
\def\compositeOperatorGdot[#1][#2]{\bar{G}_{#1}{(#2)}}

\def\setOfOperators[#1]{{\mathfrak{J}}_{#1}}

\def\operatorBinary[#1][#2]{O_{#1}^{(#2)}}
\def\operatorSign[#1][#2]{{\rm S}_{#1}^{(#2)}}
\def\operatorScaleA[#1][#2]{{\rm{A}}_{#1}^{(#2)}}
\def\operatorScaleB[#1][#2]{{\rm{B}}_{#1}^{(#2)}}
\def\operatorAngle[#1][#2]{\Omega_{#1}^{(#2)}}
\def\operatorSeparation[#1][#2]{\Delta_{#1}^{(#2)}}
\def\operatorOrderA[#1][#2]{\dot{\rm O}_{#1}^{(#2)}}
\def\operatorOrderB[#1][#2]{\ddot{\rm O}_{#1}^{(#2)}}
\def\operatorAngleScaleA[#1][#2]{\dot{\Omega}_{#1}^{(#2)}}
\def\operatorAngleScaleB[#1][#2]{\ddot{\Omega}_{#1}^{(#2)}}
\def\operatorAngleConjScaleA[#1][#2]{\dot{\Upsilon}_{#1}^{(#2)}}
\def\operatorAngleConjScaleB[#1][#2]{\ddot{\Upsilon}_{#1}^{(#2)}}

\def\functionf[#1]{p^{(#1)}}
\def\functiong[#1]{q^{(#1)}}

\def\functionfdot[#1]{\bar{p}_{\indexIterationANF}^{(#1)}}
\def\functiongdot[#1]{\bar{q}_{\indexIterationANF}^{(#1)}}

\def\funcGfForANF[#1]{f_{#1}}
\def\funcGgForANF[#1]{g_{#1}}
\def\polySeq[#1][#2]{p_{#1}(#2)}

\newcommand\mydots{\hbox to 1em{.\hss.\hss.}}

\documentclass[conference]{IEEEtran}
\IEEEoverridecommandlockouts
\usepackage{multicol}
\usepackage{lipsum}
\usepackage{mathtools}
\usepackage{amsthm}
\usepackage{acronym}
\usepackage{stfloats}
\usepackage{amsfonts}
\usepackage{acronym}
\usepackage{cite}
\usepackage{multirow}
\usepackage{bm}
\usepackage{amsmath,amssymb}
\usepackage{graphicx}
\usepackage[table,xcdraw]{xcolor}
\usepackage[english]{babel}
\usepackage[caption=false,font=footnotesize]{subfig}
\usepackage[geometry]{ifsym}
\usepackage{array}
\usepackage[utf8]{inputenc}
\usepackage[T1]{fontenc}

\usepackage{tikz}
\usepackage{lipsum}
\usetikzlibrary{calc,shadings,patterns}
\makeatletter
\tikzset{%
  remember picture with id/.style={%
    remember picture,
    overlay,
    save picture id=#1,
  },
  save picture id/.code={%
    \edef\pgf@temp{#1}%
    \immediate\write\pgfutil@auxout{%
      \noexpand\savepointas{\pgf@temp}{\pgfpictureid}}%
  },
  if picture id/.code args={#1#2#3}{%
    \@ifundefined{save@pt@#1}{%
      \pgfkeysalso{#3}%
    }{
      \pgfkeysalso{#2}%
    }
  }
}

\def\savepointas#1#2{%
  \expandafter\gdef\csname save@pt@#1\endcsname{#2}%
}

\def\tmk@labeldef#1,#2\@nil{%
  \def\tmk@label{#1}%
  \def\tmk@def{#2}%
}

\tikzdeclarecoordinatesystem{pic}{%
  \pgfutil@in@,{#1}%
  \ifpgfutil@in@%
    \tmk@labeldef#1\@nil
  \else
    \tmk@labeldef#1,(0pt,0pt)\@nil
  \fi
  \@ifundefined{save@pt@\tmk@label}{%
    \tikz@scan@one@point\pgfutil@firstofone\tmk@def
  }{%
  \pgfsys@getposition{\csname save@pt@\tmk@label\endcsname}\save@orig@pic%
  \pgfsys@getposition{\pgfpictureid}\save@this@pic%
  \pgf@process{\pgfpointorigin\save@this@pic}%
  \pgf@xa=\pgf@x
  \pgf@ya=\pgf@y
  \pgf@process{\pgfpointorigin\save@orig@pic}%
  \advance\pgf@x by -\pgf@xa
  \advance\pgf@y by -\pgf@ya
  }%
}

\makeatother

\newcounter{hatchNumber}
\DeclarePairedDelimiter\floor{\lfloor}{\rfloor}

\usepackage{cuted}
\makeatletter
\newif\ifAC@uppercase@first%
\def\Aclp#1{\AC@uppercase@firsttrue\aclp{#1}\AC@uppercase@firstfalse}%
\def\AC@aclp#1{%
	\ifcsname fn@#1@PL\endcsname%
	\ifAC@uppercase@first%
	\expandafter\expandafter\expandafter\MakeUppercase\csname fn@#1@PL\endcsname%
	\else%
	\csname fn@#1@PL\endcsname%
	\fi%
	\else%
	\AC@acl{#1}s%
	\fi%
}%
\def\Acp#1{\AC@uppercase@firsttrue\acp{#1}\AC@uppercase@firstfalse}%
\def\AC@acp#1{%
	\ifcsname fn@#1@PL\endcsname%
	\ifAC@uppercase@first%
	\expandafter\expandafter\expandafter\MakeUppercase\csname fn@#1@PL\endcsname%
	\else%
	\csname fn@#1@PL\endcsname%
	\fi%
	\else%
	\AC@ac{#1}s%
	\fi%
}%
\def\Acfp#1{\AC@uppercase@firsttrue\acfp{#1}\AC@uppercase@firstfalse}%
\def\AC@acfp#1{%
	\ifcsname fn@#1@PL\endcsname%
	\ifAC@uppercase@first%
	\expandafter\expandafter\expandafter\MakeUppercase\csname fn@#1@PL\endcsname%
	\else%
	\csname fn@#1@PL\endcsname%
	\fi%
	\else%
	\AC@acf{#1}s%
	\fi%
}%
\def\Acsp#1{\AC@uppercase@firsttrue\acsp{#1}\AC@uppercase@firstfalse}%
\def\AC@acsp#1{%
	\ifcsname fn@#1@PL\endcsname%
	\ifAC@uppercase@first%
	\expandafter\expandafter\expandafter\MakeUppercase\csname fn@#1@PL\endcsname%
	\else%
	\csname fn@#1@PL\endcsname%
	\fi%
	\else%
	\AC@acs{#1}s%
	\fi%
}%
\edef\AC@uppercase@write{\string\ifAC@uppercase@first\string\expandafter\string\MakeUppercase\string\fi\space}%
\def\AC@acrodef#1[#2]#3{%
	\@bsphack%
	\protected@write\@auxout{}{%
		\string\newacro{#1}[#2]{\AC@uppercase@write #3}%
	}\@esphack%
}%
\def\Acl#1{\AC@uppercase@firsttrue\acl{#1}\AC@uppercase@firstfalse}
\def\Acf#1{\AC@uppercase@firsttrue\acf{#1}\AC@uppercase@firstfalse}
\def\Ac#1{\AC@uppercase@firsttrue\ac{#1}\AC@uppercase@firstfalse}
\def\Acs#1{\AC@uppercase@firsttrue\acs{#1}\AC@uppercase@firstfalse}

\newtheorem{theorem}{Theorem}

\newtheorem{corollary}[theorem]{Corollary}
 
\newtheorem{example}{\color{black} Example} 

\acrodef{OOB}{out-of-band}
\acrodef{PMEPR}{peak-to-mean envelope power ratio}
\acrodef{SIC}{successive interference cancellation}
\acrodef{PAPR}{peak-to-average-power ratio}
\acrodef{APAC}{aperiodic autocorrelation}
\acrodef{OFDM}{orthogonal frequency division multiplexing}
\acrodef{DFT}{discrete Fourier transform}
\acrodef{DC}{direct current}
\acrodef{CS}{complementary sequence}
\acrodef{GCP}{Golay complementary pair}
\acrodef{ANF}{algebraic normal form}
\acrodef{PSK}{phase-shift keying}
\acrodef{QAM}{quadrature amplitude modulation}
\acrodef{QPSK}{quadrature phase-shift keying}
\acrodef{GDJ}{Golay-Davis-Jedwab}
\acrodef{FFT}{fast Fourier transform}
\acrodef{BER}{bit-error ratio}
\acrodef{SNR}{signal-to-noise ratio}
\acrodef{4G}{Fourth Generation}
\acrodef{5G}{Fifth Generation}
\acrodef{NR}{New Radio}
\acrodef{LTE}{Long-Term Evolution}
\acrodef{PTS}{partial transmit sequences}
\acrodef{PSD}{power spectral density}
\acrodef{LDPC}{low-density parity check}
\acrodef{SE}{spectral efficiency}
\acrodef{eLAA}{enhanced licensed-assisted access}
\acrodef{NR-U}{NR-Unlicensed}
\acrodef{RM}{Reed-Muller}
\acrodef{AE}{autoencoder}
\acrodef{DNN}{deep neural network}
\acrodef{OFDM-AE}{OFDM-based autoencoder}
\acrodef{DL}{deep learning}
\acrodef{CP}{cyclic prefix}
\acrodef{AWGN}{additive white Gaussian noise}
\acrodef{P2C}{polar-to-Cartesian}
\acrodef{CFR}{channel frequency response}
\acrodef{ReLU}{rectified linear unit}
\acrodef{MMSE}{minimum mean sqaure error}
\acrodef{BPSK}{binary phase-shift keying}
\acrodef{BLER}{block error rate}
\acrodef{ML}{maximum-likelihood}
\acrodef{PHY}{physical layer}
\acrodef{PA}{power amplifier}
\acrodef{IDFT}{inverse DFT}
\acrodef{DoF}{degrees-of-freedom}
\acrodef{IoT}{Internet-of-Things}
\acrodef{DFT-s-OFDM}{discrete Fourier transform-spread orthogonal frequency division multiplexing}
\acrodef{MMSE}{minimum mean square error}
\acrodef{FDE}{frequency-domain equalization}
\acrodef{FrFT}{fractional Fourier transform}
\acrodef{TF}{time-frequency}
\acrodef{BFSK}{binary frequency-shift keying}
\acrodef{CSS}{chirp-spread spectrum}
\acrodef{BCSS}{binary chirp spread spectrum}
\acrodef{EVA}{Extended Vehicular A}
\acrodef{MIMO}{multi-input multi-output}
\acrodef{PIC}{parallel interference cancellation}
\acrodef{LoRa}{Long Range}
\acrodef{HF}{high-frequency}
\acrodef{FDSS}{frequency-domain spectral shaping}
\acrodef{OCB}{occupied channel bandwidth}
\acrodef{FSK}{frequency-shift keying}
\acrodef{RF}{radio-frequency}
\acrodef{IM}{index modulation}
\acrodef{DFRC}{dual-function radar and communication}

\def\BibTeX{{\rm B\kern-.05em{\sc i\kern-.025em b}\kern-.08em
    T\kern-.1667em\lower.7ex\hbox{E}\kern-.125emX}}
\begin{document}

\title{
A Wideband Index Modulation  with Circularly-Shifted Chirps
}



\author{
\IEEEauthorblockN{Safi Shams Muhtasimul Hoque\IEEEauthorrefmark{1}, Chao-Yu Chen\IEEEauthorrefmark{2}, Alphan \c{S}ahin\IEEEauthorrefmark{1}}
\IEEEauthorblockA{\IEEEauthorrefmark{1}Electrical  Engineering Department,
University of South Carolina, Columbia, SC, USA}
\IEEEauthorblockA{\IEEEauthorrefmark{2}Department of Engineering Science, National Cheng Kung University, Tainan, Taiwan, R.O.C}
 E-mail: shoque@email.sc.edu, super@mail.ncku.edu.tw,  asahin@mailbox.sc.edu }

\maketitle

\begin{abstract}
In this study, we propose a wideband \ac{IM} based on circularly-shifted chirps. To derive the proposed method, we first prove that a \ac{GCP} can be constructed by linearly combining the Fourier series of chirps. We show that Fresnel integrals and/or Bessel functions, arising from sinusoidal and linear chirps, respectively, can lead to \acp{GCP}. We then exploit \ac{DFT-s-OFDM} to obtain a low-complexity transmitter and receiver. We also discuss its generalization for achieving a trade-off between \ac{PMEPR} and  \ac{SE}. Through comprehensive simulations, we compare the proposed scheme with \ac{DFT-s-OFDM} with \ac{IM}, \ac{OFDM} with \ac{IM} and \acp{CS} from \ac{RM} code. Our numerical results show that the proposed method limits the \ac{PMEPR} while exploiting the frequency selectivity in fading channels without an auxiliary method.
\end{abstract}


\acresetall

\section{Introduction}

\Ac{IM} is a general modulation concept which disables or activates a set of resources in different dimensions (e.g., antennas, time slots, subcarriers, codes) while providing a flexible framework for traditional coherent modulation techniques. It has been considered for \ac{5G} networks as it can provide a trade-off between energy efficiency and \ac{SE}  under an \ac{OFDM} framework and spatial techniques \cite{Ishikawa_2016, ChengIM_2018}  (see also the references therein). 

Within the last decade, \ac{OFDM} with \ac{IM} has been studied extensively, particularly for achieving a reliable high data-rate transmission. For example, in \cite{basar_2013}, the subcarriers are partitioned into several subblocks and the specific indices of the active subcarriers in each subblock are used for data transmission. In \cite{Basar2015_CIIM}, it is combined with an interleaver, where the real and imaginary parts of the complex data symbols are transmitted over different active subcarriers of the \ac{OFDM}-\ac{IM} scheme.   In \cite{Wen_2016}, interleaved versus localized grouping is studied. In \cite{dogukan2020supermode}, repetition coding over multiple subcarriers and partitioned constellations are proposed to achieve a diversity gain. Nevertheless, it has been shown that \ac{IM} is more beneficial for a low data-rate communication scenario  as compared to typical \ac{OFDM} transmission through a minimum Euclidean distance analysis \cite{Ishikawa_2016}. Given this result, in this study, we target a reliable low data-rate transmission with \ac{IM} where we can exploit the frequency selectivity naturally through wideband circularly-shifted chirp signals.

A chirp signal can sweep a large frequency spectrum over time. The efficacy of a chirp signal in a radar system was first demonstrated in \cite{darlingtonPatent1949} and later extended to the communications by encoding bits with negative and positive slopes of a linear chirp in the \ac{TF} plane. A low data-rate communication through a proprietary chirp modulation, called \ac{LoRa}, was introduced by Samtec for \ac{IoT} networks. To achieve a higher data rate, a chirp signal is translated in the frequency domain to achieve a set of orthogonal chirps \cite{fresnel_nozh_7523229}. In \cite{sahin_2020}, circularly-shifted chirps are proposed by using the structure of \ac{DFT-s-OFDM} with a specific \ac{FDSS} function. Hence,  a framework is established for chirps by using the typical \ac{OFDM} blocks. 

In this study, we consider circularly-shifted chirps with \ac{IM} and exploit the framework in \cite{sahin_2020}. We provide both theoretical and practical contributions in various areas listed as follows:
\begin{itemize}
    \item {\bf Chirps and \acp{CS}}: We prove that a \ac{CS} \cite{Golay_1961} can be constructed by linearly combining the Fourier series of the chirp signals. Hence, two seemingly different topics are shown to be connected.
    \item {\bf Low-complexity framework:} We show that the \ac{DFT-s-OFDM} can be utilized to achieve a low-complexity wideband \ac{IM} with chirps. We also generalize the proposed scheme to obtain a trade-off between \ac{PMEPR} and \ac{SE}.
    \item {\bf Diversity gain:} We show that the wideband nature of chirps exploits the frequency selectivity without any auxiliary method.
    \item {\bf Comprehensive comparisons:} We compare the proposed scheme with \ac{OFDM}-\ac{IM}, \ac{DFT-s-OFDM}-\ac{IM}, and the \ac{CS} from \ac{RM} codes \cite{davis_1999,Sahin_20twc}, numerically, in terms of \ac{PMEPR} and  error rate.
\end{itemize}

The rest of the paper is organized as follows. In Section \ref{sec:prelim},  preliminary discussions  are provided.  In Section \ref{sec:connection}, the relationship between \acp{CS} and chirps are discussed. By using the results obtained in Section \ref{sec:connection}, a low-complexity modulation scheme is derived. In Section \ref{sec:numresults}, the proposed scheme is assessed numerically. The paper is concluded in Section \ref{sec:conclusion}.

{\em Notation:} The sets of complex numbers and positive integers are denoted by $\complexNumbers$ and $\integersPositiveSet$, respectively. Conjugation is denoted by $(\cdot)^*$. The notation $(\eleGa[{\indexEleOfSeq}])_{i=0}^{\lengthGaGb-1}$ represents the sequence $\seqGa= (\eleGa[0],\eleGa[1],\dots, \eleGa[\lengthGaGb-1])$. The constant $\constantj$ denotes $\sqrt{-1}$.

\section{Preliminaries}
\label{sec:prelim}

An \ac{OFDM} symbol with the symbol duration $\symbolDuration$ can be expressed in continuous time as a polynomial given by
\begin{align}
\polySeq[\seqGa][\polyVariable] \triangleq \eleGa[\numberOfShifts-1]\polyVariable^{\numberOfShifts-1} + \eleGa[\numberOfShifts-2]\polyVariable^{\numberOfShifts-2}+ \dots + \eleGa[0]~,
\label{eq:polySeq}
\end{align}
where $\seqGa = (\eleGa[0],\eleGa[1],\mydots,\eleGa[\numberOfShifts])$ is a sequence of length $\lengthGaGb$, and $\polyVariable\in\{\constante^{\constantj\frac{2\pi\timeVar}{\symbolDuration}}| 0\le\timeVar <\symbolDuration
\}$.

\subsection{Circularly-shifted Chirps}
Let $\basisFunction[{\amountOfShift[\indexTime]}]=\constante^{\constantj\angleSignal[\indexTime]}$ be the $\indexTime$th circular translation of an  arbitrary band-limited function with the duration  $\symbolDuration$, where $\amountOfShift[\indexTime]$ is the amount of circular shift and $m=0,1,2,\mydots, \numberOfShifts-1$. By using Fourier series,  $\basisFunction[{\amountOfShift[\indexTime]}]$ can be approximately expressed as
\begin{align}
\basisFunction[{\amountOfShift[\indexTime]}]  \approx \sum_{\indexSubcarrier=\lowerFrequency}^{\upperFrequency} \fourierSeries[\indexSubcarrier] \constante^{\constantj2\pi\indexSubcarrier\frac{\timeVar-\amountOfShift[\indexTime]}{\symbolDuration}}~,
\label{eq:basisDecompose}
\end{align}
where $\lowerFrequency<0$ and $\upperFrequency>0$ are integers, and $\fourierSeries[\indexSubcarrier]$ is the $\indexSubcarrier$th Fourier coefficient given by
\begin{align}
\fourierSeries[\indexSubcarrier] = \fourier[{\constante^{\constantj\angleSignal[0]}}] \triangleq \frac{1}{\symbolDuration}\int_{\symbolDuration}
\constante^{\constantj\angleSignal[0]}
\constante^{-\constantj2\pi\indexSubcarrier\frac{\timeVar}{\symbolDuration}}d\timeVar. 
\end{align}

Let ${\numberOfOccupiedSubcarriers}/{2\symbolDuration}$ be the maximum frequency deviation around the carrier frequency. The approximation in \eqref{eq:basisDecompose} then becomes more accurate for  $\lowerFrequency<-\numberOfOccupiedSubcarriers/2$ and $\upperFrequency>\numberOfOccupiedSubcarriers/2$. This is due to the fact that $\basisFunction[{\amountOfShift[\indexTime]}]$ is a band-limited function, i.e.,  $|\fourierSeries[\indexSubcarrier]|$  approaches zero rapidly for $|\indexSubcarrier|>{\numberOfOccupiedSubcarriers}/{2}$. Note that the actual bandwidth of a chirp  is slightly larger than twice the maximum frequency deviation \cite{proakisfundamentals}. It can be calculated based on the total integrated power of the transmitted spectrum, i.e., \ac{OCB}. In this study, we express the \ac{OCB} as $\OCB/\symbolDuration$ Hz where $\OCB\in\integersPositiveSet$. Also, the
instantaneous frequency of $\basisFunction[{\amountOfShift[\indexTime]}]$ around the carrier frequency $\fcarrier$ can  be obtained as $\instantaneousFrequency[\indexTime]=\frac{1}{2\pi}d{\angleSignal[\indexTime]}/{d\timeVar}$ Hz. 
 
\subsubsection{Linear Chirps}
\label{subsec:linearChirp}
Let $\instantaneousFrequency[0]$ be a linear function changing from $-\frac{\numberOfOccupiedSubcarriers}{2\symbolDuration}$ Hz to $\frac{\numberOfOccupiedSubcarriers}{2\symbolDuration}$ Hz, i.e., $\instantaneousFrequency[0]=\frac{\numberOfOccupiedSubcarriers}{2\symbolDuration}\left(\frac{2\timeVar}{\symbolDuration}-1\right)$. The $\indexSubcarrier$th Fourier coefficient can be calculated as 
\begin{align}
\fourierSeries[\indexSubcarrier] = \linearCoef(\fresnelC[{\linearXone}] + \fresnelC[{\linearXtwo}] + \constantj\fresnelS[{\linearXone}] + \constantj\fresnelS[{\linearXtwo}])~,
\label{eq:fresnekfcn}
\end{align}
where $\fresnelC[{\cdot}]$ and  $\fresnelS[{\cdot}]$ are the Fresnel integrals with cosine and sine functions, respectively, and $\linearXone=(\numberOfOccupiedSubcarriers/2+2\pi\indexSubcarrier)/\sqrt{\pi\numberOfOccupiedSubcarriers }$, $\linearXtwo=(\numberOfOccupiedSubcarriers/2-2\pi\indexSubcarrier)/\sqrt{\pi\numberOfOccupiedSubcarriers}$, and $\linearCoef= \sqrt{\frac{\pi}{\numberOfOccupiedSubcarriers}}\constante^{-\constantj\frac{(2\pi\indexSubcarrier)^2}{2\numberOfOccupiedSubcarriers}-\constantj\pi \indexSubcarrier}$ \cite{sahin_2020}.

\subsubsection{Sinusoidal Chirps}
\label{subsec:sinosoidalChirp}
Let  $\instantaneousFrequency[0]=\frac{\numberOfOccupiedSubcarriers}{2\symbolDuration}\cos\left({2\pi \frac{\timeVar}{\symbolDuration}}\right)$. In this case,  the maximum frequency deviation is $\numberOfOccupiedSubcarriers/2\symbolDuration$ Hz and it can be shown that 
\begin{align}
    \fourierSeries[\indexSubcarrier] = \besselFunctionFirstKind[\indexSubcarrier][\frac{\numberOfOccupiedSubcarriers}{2}],
    \label{eq:besselfcn}
\end{align}
where $\besselFunctionFirstKind[\indexSubcarrier][\cdot]$ is the Bessel function of the first kind of order $\indexSubcarrier$ \cite{proakisfundamentals}.
 
\subsection{Chirps with DFT-s-OFDM}
\label{subsec:dftsofdmchirp}
Let complex baseband signal  $\transmittedSignal[\timeVar]$ be a linear combinations of the translated chirps as 
\begin{align}
\transmittedSignal[\timeVar] = \sum_{\indexTime=0}^{\numberOfShifts-1} \dataSymbols[\indexTime] \basisFunction[{\amountOfShift[\indexTime]}]
~,
\label{eq:originalWaveform}
\end{align}
where $\dataSymbols[\indexTime]\in\complexNumbers$ is the $\indexTime$th modulation symbol, e.g., a \ac{PSK} symbol. In \cite{sahin_2020}, it was shown that if  $\amountOfShift[\indexTime]=\indexTime/\numberOfShifts\times\symbolDuration$ (i.e., uniform spacing in time), $\transmittedSignal[\timeVar] $ in discrete time can be written  as
\begin{align}
\transmittedSignal[\frac{\indexSample\symbolDuration}{\idftSize}] \approx&\sum_{\indexSubcarrier=\lowerFrequency}^{\upperFrequency}{\fourierSeries[\indexSubcarrier]{\sum_{\indexTime=0}^{\numberOfShifts-1} \dataSymbols[\indexTime] 
			\constante^{-\constantj2\pi \indexSubcarrier \frac{\indexTime}{\numberOfShifts}}}
	\constante^{\constantj2\pi \indexSubcarrier \frac{\indexSample}{\idftSize}}}~,
\label{eq:chirpWave}
\end{align}
by sampling $\transmittedSignal[\timeVar]$ with the period  $\symbolDuration/\idftSize$. Equation \eqref{eq:chirpWave} shows that chirps signal can be synthesized by using the  \ac{DFT-s-OFDM} transmitter with a special choice of \ac{FDSS} sequence which substantially reduces the transmitter complexity. Since a typical \ac{DFT-s-OFDM} receiver with a single-tap \ac{MMSE} \ac{FDE} can be utilized \cite{Sari_1995}, it also offers a low-complexity at the receiver side. Note that the condition $\numberOfShifts\ge\OCB>\numberOfOccupiedSubcarriers$ must hold to be able to represent a chirp by using $\numberOfShifts$ subcarriers. In this study, we also consider that  $\upperFrequency-\lowerFrequency+1=\numberOfShifts$ to avoid \ac{FDSS} beyond $\numberOfShifts$ subcarriers. 

\subsection{Complementary Sequences}
\label{subsec:complementarySequence}
The sequence pair  $(\seqGa,\seqGb)$ of length  $\lengthGaGb$ is a \ac{GCP} if $\apac[\seqGa][\lagForCorrelation]+\apac[\seqGb][\lagForCorrelation] = 0$ for $\lagForCorrelation~\neq0~$ where $\apac[\seqGa][\lagForCorrelation]$ and $\apac[\seqGb][\lagForCorrelation]$ are the \acp{APAC} of the sequences $\seqGa$ and $\seqGb$ at the $\lagForCorrelation$th lag, respectively \cite{Golay_1961}. Each sequence in a \ac{GCP} is called a \ac{CS}. A GCP $(\seqGa,\seqGb)$ can also be defined as any sequence pair satisfying
${|\polySeq[\seqGa][{\constante^{\constantj\frac{2\pi\timeVar}{\symbolDuration}}}]|^2+|\polySeq[\seqGb][{\constante^{\constantj\frac{2\pi\timeVar}{\symbolDuration}}}]|^2} ={\apac[\seqGa][0]+\apac[\seqGb][0]}$ \cite{parker_2003}.

If a \ac{CS} is transmitted with an \ac{OFDM} waveform, the instantaneous peak power of the corresponding signal is bounded, i.e., 
$\max_{\timeVar}|\polySeq[\seqGa][{
	\constante^{\constantj\frac{2\pi\timeVar}{\symbolDuration}}
}]|^2 \le \apac[\seqGa][0]+\apac[\seqGb][0]$. As a result, the \ac{PMEPR} of the \ac{OFDM} symbol $\polySeq[\seqGa][{
\constante^{\constantj\frac{2\pi\timeVar}{\symbolDuration}}}]$, defined as $\max_{\timeVar}|\polySeq[\seqGa][{
	\constante^{\constantj\frac{2\pi\timeVar}{\symbolDuration}}
}]|^2/\averagePower$,  is  less than or equal to $10\log_{10}(2)\approx3$~dB if $\averagePower=\apac[\seqGa][0]=\apac[\seqGb][0]$. Note that for non-unimodular \acp{CS}, $\apac[\seqGa][0]$ may not be equal to $\apac[\seqGb][0]$. In that case, the power of \ac{OFDM} symbol with $\seqGa$ can be different from the one  with $\seqGb$ although the instantaneous peak power is still less than  or equal to $\apac[\seqGa][0]+\apac[\seqGb][0]$ for both symbols. Hence, to avoid misleading results, we define $\averagePower$ as the power of the complex baseband signal in this study.

\section{Chirp-based Complementary Sequences}
\label{sec:connection}
\begin{theorem}
\label{th:golaypair}
Let $\seqx[\timeVar]$ and $\seqy[\timeVar]$ be two complex-valued signals defined by
\begin{align}
\seqx[\timeVar] = \dataSymbols[\chirpm]\constante^{\constantj\chirpphase[\chirpm][\timeVar]} + \dataSymbols[\chirpn]\constante^{\constantj\chirpphase[\chirpn][\timeVar]}~,
\label{eq:Golay1}
\\
\seqy[\timeVar] = \dataSymbols[\chirpm]\constante^{\constantj\chirpphase[\chirpm][\timeVar]} - \dataSymbols[\chirpn]\constante^{\constantj\chirpphase[\chirpn][\timeVar]}~,
\label{eq:Golay2}
\end{align}
where $\dataSymbols[\chirpm],\dataSymbols[\chirpn]\in\complexNumbers$ and $|\dataSymbols[\chirpm]|=|\dataSymbols[\chirpn]|=1$. The Fourier coefficients of $\seqx[\timeVar]$ and $\seqy[\timeVar]$ form a \ac{GCP}.
\end{theorem}
\begin{proof}
By the definition, we need to show that $|\seqx[\timeVar]|^2+|\seqy[\timeVar]|^2$ is constant:
\begin{align}
|\seqx[\timeVar]|^2 
=&|\dataSymbols[\chirpm]|^2+|\dataSymbols[\chirpn]|^2\nonumber\\
&+\dataSymbols[\chirpm]\dataSymbols[\chirpn]^*\constante^{\constantj(\chirpphase[\chirpm][\timeVar]-\chirpphase[\chirpn][\timeVar])}+\dataSymbols[\chirpm]^*\dataSymbols[\chirpn]\constante^{-\constantj(\chirpphase[\chirpm][\timeVar]-\chirpphase[\chirpn][\timeVar])}~.\nonumber
\end{align}
Similarly,
\begin{align}
|\seqy[\timeVar]|^2
=&|\dataSymbols[\chirpm]|^2+|\dataSymbols[\chirpn]|^2\nonumber\\
&-\dataSymbols[\chirpm]\dataSymbols[\chirpn]^*\constante^{\constantj(\chirpphase[\chirpm][\timeVar]-\chirpphase[\chirpn][\timeVar])}-\dataSymbols[\chirpm]^*\dataSymbols[\chirpn]\constante^{-\constantj(\chirpphase[\chirpm][\timeVar]-\chirpphase[\chirpn][\timeVar])}~.\nonumber
\end{align}
Therefore, $|\seqx[\timeVar]|^2+|\seqy[\timeVar]|^2
=2\times(|\dataSymbols[\chirpm]|^2+|\dataSymbols[\chirpn]|^2)
=4$,
which implies that $\fourier[{\seqx[{\timeVar}]}]$ and  $\fourier[{\seqy[{\timeVar}]}]$ form a \ac{GCP}.
\end{proof}
Theorem~\ref{th:golaypair} indicates that the Fourier coefficients of a linear combination of the frequency response of two constant-envelope chirps result in a \ac{CS}. As a result, it yields an interesting connection between chirps and \acp{CS}.
\begin{example}
\rm 
Assume that $\seqx[\timeVar]$ and $\seqy[\timeVar]$ are  linear combinations of two circularly-shifted versions of  a band-limited sinusoidal chirp defined in Section~\ref{subsec:sinosoidalChirp}. Therefore, by using \eqref{eq:basisDecompose} and \eqref{eq:besselfcn}, the Fourier coefficients of $\seqx[\timeVar]$ and $\seqy[\timeVar]$ are obtained as
\begin{align}
    \eleGa[\indexSubcarrier] &= \dataSymbols[\chirpm]\besselFunctionFirstKind[\indexSubcarrier][\frac{\numberOfOccupiedSubcarriers}{2}]\constante^{-\constantj2\pi\indexSubcarrier\frac{\amountOfShift[\chirpm]}{\symbolDuration}}+\dataSymbols[\chirpn]\besselFunctionFirstKind[\indexSubcarrier][\frac{\numberOfOccupiedSubcarriers}{2}]\constante^{-\constantj2\pi\indexSubcarrier\frac{\amountOfShift[\chirpn]}{\symbolDuration}}~, 
    \label{eq:CSbesselA}
    \\
    \eleGb[\indexSubcarrier] &= \dataSymbols[\chirpm]\besselFunctionFirstKind[\indexSubcarrier][\frac{\numberOfOccupiedSubcarriers}{2}]\constante^{-\constantj2\pi\indexSubcarrier\frac{\amountOfShift[\chirpm]}{\symbolDuration}}-\dataSymbols[\chirpn]\besselFunctionFirstKind[\indexSubcarrier][\frac{\numberOfOccupiedSubcarriers}{2}]\constante^{-\constantj2\pi\indexSubcarrier\frac{\amountOfShift[\chirpn]}{\symbolDuration}}\label{eq:CSbesselB},
\end{align}
respectively. Based on Theorem~\ref{th:golaypair}, $(\eleGa[{\indexEleOfSeq}])_{i=-\infty}^{\infty}$ and $(\eleGb[{\indexEleOfSeq}])_{i=-\infty}^{\infty}$ form a \ac{GCP}. Since the sinusoidal chirps are band-limited signals, the amplitude of a Fourier coefficient approaches to zero for $|\indexEleOfSeq|\ge\numberOfOccupiedSubcarriers/2$. Therefore, $(\eleGa[{\indexEleOfSeq}])_{i=\lowerFrequency}^{\upperFrequency}$ and $(\eleGb[{\indexEleOfSeq}])_{i=\lowerFrequency}^{\upperFrequency}$ are approximately \ac{GCP}. Note that if the sinusoidal chirps are replaced by the linear chirps, by using \eqref{eq:fresnekfcn},  the Fourier coefficients of $\seqx[\timeVar]$ and $\seqy[\timeVar]$ can be calculated as
\begin{align}
    \eleGa[\indexSubcarrier] =& \dataSymbols[\chirpm]
    \linearCoef(\fresnelC[{\linearXone}] + \fresnelC[{\linearXtwo}] + \constantj\fresnelS[{\linearXone}] + \constantj\fresnelS[{\linearXtwo}])
    \constante^{-\constantj2\pi\indexSubcarrier\frac{\amountOfShift[\chirpm]}{\symbolDuration}}\nonumber\\ &+\dataSymbols[\chirpn]\linearCoef(\fresnelC[{\linearXone}] + \fresnelC[{\linearXtwo}] + \constantj\fresnelS[{\linearXone}] + \constantj\fresnelS[{\linearXtwo}])\constante^{-\constantj2\pi\indexSubcarrier\frac{\amountOfShift[\chirpn]}{\symbolDuration}}~,\nonumber
    \\
    \eleGb[\indexSubcarrier] =& \dataSymbols[\chirpm]
    \linearCoef(\fresnelC[{\linearXone}] + \fresnelC[{\linearXtwo}] + \constantj\fresnelS[{\linearXone}] + \constantj\fresnelS[{\linearXtwo}])
    \constante^{-\constantj2\pi\indexSubcarrier\frac{\amountOfShift[\chirpm]}{\symbolDuration}}\nonumber\\ &-\dataSymbols[\chirpn]\linearCoef(\fresnelC[{\linearXone}] + \fresnelC[{\linearXtwo}] + \constantj\fresnelS[{\linearXone}] + \constantj\fresnelS[{\linearXtwo}])\constante^{-\constantj2\pi\indexSubcarrier\frac{\amountOfShift[\chirpn]}{\symbolDuration}}~.\nonumber
\end{align}
\end{example}

In \figurename~\ref{fig:correlationExample}, we exemplify a \ac{GCP} of length $\numberOfShifts=24$, synthesized through \eqref{eq:CSbesselA} and \eqref{eq:CSbesselB} for $\lowerFrequency=-11$, $\upperFrequency=12$,  ${\amountOfShift[\chirpm]}/{\symbolDuration}=0/24$, ${\amountOfShift[\chirpn]}/{\symbolDuration}=1/24$, and $\dataSymbols[\chirpm]=\dataSymbols[\chirpn]=1$. When $\numberOfOccupiedSubcarriers=24$, the sequences are truncated heavily. Therefore, it does not satisfy the definition of a \ac{GCP} given in Section~\ref{subsec:complementarySequence}. On the other hand, when the maximum frequency deviation is halved, $\OCB$ is $15$ for containing $99\%$  of the total integrated power of the  spectrum. Hence, $\numberOfShifts=24$ forms the chirps well and the resulting sequences form a \ac{GCP}. It is also worth noting that synthesized \acp{CS} are not unimodular sequences. Therefore, the mean power of \ac{OFDM} symbol changes although  instantaneous power is bounded.

\begin{figure}[t]
	\centering
	{\includegraphics[width =3.3in]{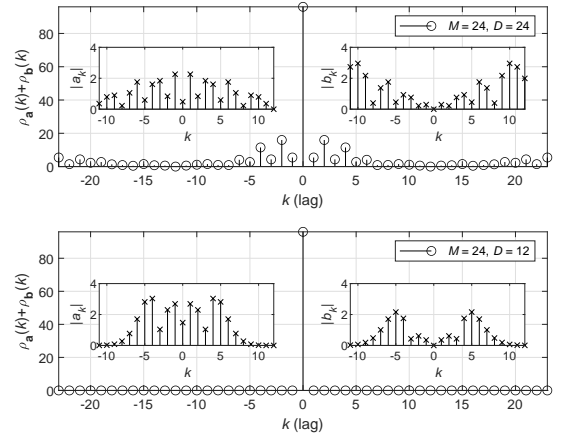}
	} 
	\vspace{-4mm}
	\caption{An instance of \ac{GCP} of length $\numberOfShifts=24$ synthesized through sinusoidal chirps. 
	}
	\label{fig:correlationExample}
\end{figure}

\begin{corollary} 
\label{co:indexmodulation}
 Let the coefficients $\dataSymbols[\chirpm],\dataSymbols[\chirpn]\in\setPSKsymbols[{\psksize}]\triangleq\{\constante^{\constantj 2\pi\times \frac{\indexpsksymbol}{\psksize }}| \indexpsksymbol=0,1,\mydots, \psksize-1\}$. Let $\setChirps\triangleq\{\basisFunction[{\amountOfShift[\indexTime]}]|\indexTime=0,1,\mydots,\numberOfShifts-1\}$ be a set of $\numberOfShifts$ uniformly circularly-shifted chirps of an  arbitrary band-limited function with the duration $\symbolDuration$ and $\OCB\le\numberOfShifts$. Without using the same chirp twice, the total number of distinct \acp{CS} of length  $\numberOfShifts$ is ${\numberOfShifts \choose 2}\times \psksize^2$.
\end{corollary} 

\begin{proof}
Since $|\setPSKsymbols[{\psksize}]|=\psksize$, there exist $\psksize^2$ combinations for $\{\dataSymbols[\chirpm],\dataSymbols[\chirpn]\}$.
As $|\setChirps|=\numberOfShifts$, $\constante^{\constantj\chirpphase[\chirpm][\timeVar]}$ and $\constante^{\constantj\chirpphase[\chirpn][\timeVar]}$ in \eqref{eq:Golay1} can be chosen in ${\numberOfShifts \choose 2}$ ways without using the same chirp.
Thus, the total number of \acp{CS} is ${\numberOfShifts \choose 2}\times \psksize^2$ via Theorem~\ref{th:golaypair}. The \acp{CS} are distinct as $\basisFunction[{\amountOfShift[\indexTime]}]$ are distinct. Since the \ac{OCB} of $\basisFunction[{\amountOfShift[\indexTime]}]$ is less than or equal to $\numberOfShifts/\symbolDuration$, the length of the synthesized \ac{CS} is $\numberOfShifts$ based on Nyquist's sampling theorem.
\end{proof}

\section{A Low-Complexity Scheme with DFT-s-OFDM}
\label{sec:scheme}
In this section, we utilize Corollary~\ref{co:indexmodulation} to develop a wideband \ac{IM} by using the structure of \ac{DFT-s-OFDM} discussed in Section~\ref{subsec:dftsofdmchirp}. At the transmitter, we consider $\numberOfInfoBits=\floor{\log_{2}({{\numberOfShifts \choose 2}\times \psksize^2})}$ information bits. Assuming $\psksize$ is an integer power-of-two, we map $2\log_{2}(\psksize)$ of the information bits to $\symbolPSK[1]$ and  $\symbolPSK[2]$, where  $\symbolPSK[1],\symbolPSK[2]\in \setPSKsymbols[{\psksize}]$.
The rest of the information bits are utilized to choose a set $\{\chirpm,\chirpn\}$, where $\chirpm,\chirpn\in\{0,1,\mydots,\numberOfShifts-1\}$ and $\chirpm\neq\chirpn$. We generate the modulation symbols as $\dataSymbols[\chirpm]=\symbolPSK[1]$, $\dataSymbols[\chirpn]=\symbolPSK[2]$, and $\dataSymbols[i|i\in\{0,1,\mydots,\numberOfShifts-1\},i\neq\chirpm,\chirpn]=0$.  We then synthesize the transmitted signal according to \eqref{eq:chirpWave}. 

\begin{figure*}
	\centering
	{\includegraphics[width =7in]{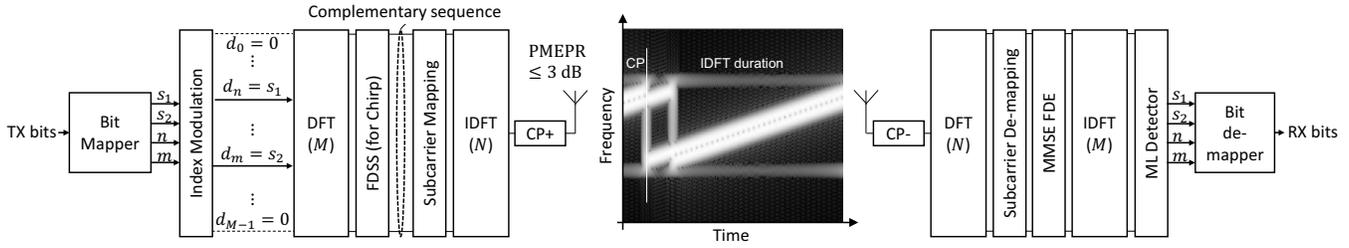}
	} 
	\caption{Transmitter and receiver block diagrams for the proposed scheme.} 
	\label{fig:txrxa}
\end{figure*}
In \figurename~\ref{fig:txrxa},  we give the transmitter and receiver block diagrams. First, $\numberOfShifts$-point \ac{DFT} of the modulation symbols $(\dataSymbols[0],\dataSymbols[1],\mydots,\dataSymbols[\numberOfShifts-1])$ is calculated. The resulting sequence is then shaped in the frequency domain with an \ac{FDSS} for generating chirps. According to Corollary~\ref{co:indexmodulation}, after \ac{FDSS}, a \ac{CS} is formed since there are only two modulation symbols for the proposed scheme. We then map the resulting \ac{CS} to the \ac{OFDM} subcarriers. After an $\idftSize$-point \ac{IDFT} of the mapped \ac{CS}, the resulting signal is transmitted with a \ac{CP}. 

At the receiver side, we consider a single-tap \ac{MMSE}-\ac{FDE} and a \ac{ML} detector for detecting $\chirpm$,  $\chirpn$, $\symbolPSK[1]$, and  $\symbolPSK[2]$.
Let $(\dataSymbolAfterIDFTspread[0],\dataSymbolAfterIDFTspread[1],\mydots,\dataSymbolAfterIDFTspread[\numberOfShifts-1])$ be the received modulation symbols after $\numberOfShifts$-point \ac{IDFT} operation. Since $|\symbolPSK[1]|=|\symbolPSK[2]|=1$, the \ac{ML} detector for $\chirpm$, $\chirpn$, $\symbolPSK[1]$, and  $\symbolPSK[2]$ can be expressed as
\begin{align}
    \{ \{\chirpmdetect, \chirpndetect\}, \symbolPSKdetect[1], \symbolPSKdetect[2]\} = \arg\max_{\substack{\{\chirpm, \chirpn\}, \symbolPSK[1], \symbolPSK[2]\\ \chirpm\neq\chirpn}} \Re\{  \dataSymbolAfterIDFTspread[\chirpm]\symbolPSK[1]^*+\dataSymbolAfterIDFTspread[\chirpn]\symbolPSK[2]^*\}~.
    \label{eq:mldet}
\end{align}
As $\chirpm\neq\chirpn$, a low-complexity \ac{ML} detector can be implemented by evaluating  $\test[i,\indexpsksymbol]=\Re\{ \dataSymbolAfterIDFTspread[i]\constante^{-\constantj2\pi\indexpsksymbol/\psksize}\}$ for $i=0,1,\mydots,\numberOfShifts-1$  and $\indexpsksymbol=0,1,\mydots,\psksize-1$ and choosing two different indices for $i$ and the corresponding $\indexpsksymbol$'s that maximize $\test[i,\indexpsksymbol]$. 

\subsection{Generalizations and Practical Issues}
For a practical radio, the mapping from the information bits to a combination of $\{\chirpm,\chirpn\}$ at the transmitter (and vice versa for the receiver) may be a challenge. This can be addressed by constructing a bijection function from integers to the set of combinations via a combinatorial number system of degree $2$ \cite{basar_2013}, also called combinadics. 

The choice of \ac{FDSS} is important for obtaining a good \ac{BER} performance and a low \ac{PMEPR}. For example, a linear chirp offers a more flat \ac{FDSS} as compared to the one with a sinusoidal chirp. In \cite{sahin_2020}, it was demonstrated that a flatter \ac{FDSS} improves the \ac{BER} performance for the receiver with a single-tap \ac{MMSE}-\ac{FDE}. On the other hand, a linear chirp causes abrupt instantaneous frequency changes within the \ac{IDFT} duration. Therefore, it requires a much lower $\numberOfOccupiedSubcarriers$ as compared to the one for a sinusoidal chirp for a given $\numberOfShifts$. This issue can distort the targeted \acp{CS} and cause a larger \ac{PMEPR} than the theoretical bound, as demonstrated in Section~\ref{sec:numresults}.

If more than two indices are allowed to be utilized, the proposed scheme can be generalized to a scheme that offers a trade-off between maximum \ac{PMEPR} and \ac{SE}.  Let $\numberofIndices$ denote the number of indices that are allowed to be used. The \ac{SE} of the scheme can be calculated as
$
    \spectralEfficiency = \frac{\floor{\log_{2}({{\numberOfShifts \choose \numberofIndices}\times \psksize^\numberofIndices})}}{\numberOfShifts}~\text{bit/second/Hz}~,
$
while the \ac{PMEPR} of a signal is always less than or equal to $\numberofIndices$ as $\numberofIndices$ chirps are transmitted simultaneously. In \figurename~\ref{fig:tradeoff}, we illustrate the trade-off between the maximum \ac{PMEPR} and the \ac{SE}  for a given $\numberOfShifts$ and $\psksize=4$ by changing $\numberofIndices$ from $1$ to $11$. The generalized scheme reduces to the scheme  in \figurename~\ref{fig:txrxa} with \acp{CS} for $\numberofIndices=2$. If there is room for more  transmit power, $\numberofIndices>2$ can be utilized for increasing the \ac{SE} while still limiting \ac{PMEPR}.

\begin{figure}[t]
	\centering
	{\includegraphics[width =3.3in]{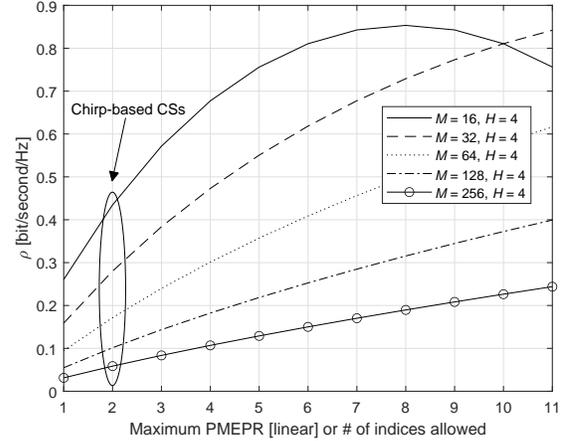}
	} 
	\vspace{-2mm}
	\caption{Trade-off between \ac{SE} and maximum \ac{PMEPR}.}
	\label{fig:tradeoff}
\end{figure}

\subsection{Comparisons}
As compared to \ac{DFT-s-OFDM}-\ac{IM}, the proposed scheme has a \ac{PMEPR} advantage since  the signal is spread in the time domain whereas \ac{DFT-s-OFDM} generates sinc pulses. The proposed scheme and \ac{OFDM} with \ac{IM} have similar \ac{PMEPR} characteristics since \ac{OFDM} with \ac{IM} also spreads the symbol energy in time. However, the energy is also distributed  within the signal bandwidth with the proposed scheme. Thus, the proposed scheme allows a coherent receiver to exploit frequency diversity in selective channels naturally. On the other hand, \ac{OFDM}-\ac{IM} receiver does not fully benefit from the frequency selectivity without an extra operation, e.g., repetitions or interleaving, at the transmitter.

In \cite{davis_1999}, a low \ac{PMEPR} coding scheme was proposed to generate \acp{CS} from \ac{RM} code. This scheme synthesizes $\numberOfPointsForPSK^{\positiveInt+1}\cdot \positiveInt!/2$ \acp{CS} where the length of each \ac{CS}  has to be in the form of $2^\positiveInt$, where $\positiveInt\in\integersPositiveSet$. 
When a seed \ac{GCP} of length $N$ is utilized with this scheme,  it can be shown that $\numberOfPointsForPSK^{\positiveInt+1}\cdot \positiveInt!$ \acp{CS} of length $N\cdot 2^\positiveInt$ can be generated \cite{Sahin_20twc}. Therefore, the spectral efficiency of the schemes in  \cite{davis_1999} and \cite{Sahin_20twc} can be calculated as $\floor{\log_2 ({ \numberOfPointsForPSK^{\positiveInt+1}\cdot \positiveInt!})}/2^{\positiveInt+1}$ and $\floor{\log_2 ({ \numberOfPointsForPSK^{\positiveInt+1}\cdot \positiveInt!})}/(N2^{\positiveInt})$ bit/second/Hz, respectively. 

The differences between these schemes and the proposed scheme can be listed as follows: 1) The proposed scheme allows the length of \ac{CS} to be chosen arbitrarily. For example, $\numberOfShifts$ can be an integer chosen as an integer multiple of 12 based on the resource allocation in 3GPP \ac{5G} \ac{NR} and \ac{4G} LTE or the resource units in IEEE 802.11ax, e.g., 26 subcarriers. 2) The schemes in \cite{davis_1999, Sahin_20twc} do not offer a trade-off between \ac{PMEPR} and spectral efficiency whereas  $\numberofIndices$ can be chosen for a higher \ac{SE} at a cost of high \ac{PMEPR} with our scheme. The \ac{PMEPR} is still theoretically limited. 3) The number of \acp{CS} is a function of a second order coset term generated through permutations  in \cite{davis_1999} and \cite{Sahin_20twc}. However, designing a decoder for all possible permutations is not trivial \cite{Paterson2000decode}. For a fixed coset, the decoder can be implemented through fast Hadamard transformation or recursive methods \cite{Schmidt2005}, but the \ac{SE} reduces to $\floor{\log_2 { \numberOfPointsForPSK^{m+1}}}$. Under this case, the \ac{SE} of the proposed scheme and the schemes in \cite{davis_1999} and \cite{Sahin_20twc} are similar while a simple decoder can be employed for the proposed method.

Note that  the \ac{SE} of the proposed scheme can be considered as low as compared to typical coding schemes such as \ac{LDPC} or polar codes. Although this appears as a disadvantage, there exist communication scenarios (e.g., uplink control channels in \ac{5G} \ac{NR} \cite{Sahin_20twc}, \ac{IoT} networks, and \ac{DFRC} scenarios) where the primary concern is reliability under low \ac{SNR} or a longer battery life, rather than a high data rate. For these scenarios, the proposed scheme provides a way of limiting \ac{PMEPR} without an optimization procedure at the transmitter while exploiting frequency selectivity.

\section{Numerical Results}
\label{sec:numresults}
For computer simulations, the symbol duration and the \ac{CP} duration are set to $\symbolDuration = 16.67~\mu$s and $\CPDuration=2.34~\mu$s, respectively, based on the  \ac{5G} \ac{NR} waveform parameters.
We assume that the transmitter uses 32 resource blocks, i.e.,  $\numberOfShifts=384$ subcarriers are used. The data symbols were generated based on \ac{QPSK}, i.e., $\psksize=4$.  For the fading channel, we consider ITU \ac{EVA}. At the transmitter, the combinations $\{\chirpm,\chirpn\}$ are generated by mapping natural order of integers to combinations and we choose $\numberOfOccupiedSubcarriers=300$. At the receiver, all possible $\{\chirpm,\chirpn\}$ combinations are considered  and a non-valid combination (e.g., due to the noise) is mapped to a valid combination to avoid catastrophic results. The \ac{FDSS} is chosen based on \eqref{eq:besselfcn} for sinusoidal chirps and \eqref{eq:fresnekfcn}  for linear chirps. We assume that the channel and \ac{FDSS} information are available at the receiver. We compare the proposed scheme with \ac{OFDM}-\ac{IM}, \ac{DFT-s-OFDM}-\ac{IM} (i.e., no \ac{FDSS} is applied),  and the \acp{CS} from \ac{RM} codes \cite{davis_1999,Sahin_20twc}. For \ac{OFDM}-\ac{IM}, an \ac{ML} detector which incorporates the channel frequency response is utilized \cite{basar_2013}.
For \acp{CS} from \ac{RM} codes, we use a seed \ac{GCP} of length $N=3$ and use $\positiveInt=7$. To achieve the same spectral efficiency, we consider $16$ permutations for $4$ extra bits transmission and   the \ac{ML} decoder  proposed in \cite{Schmidt2005} (channel frequency response is also incorporated) is employed  for each permutation.

\begin{figure}[t]
	\centering
	{\includegraphics[width =3.3in]{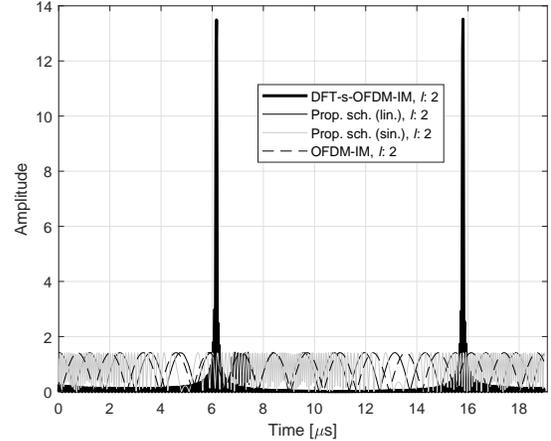}
	} 
	\vspace{-2mm}
	\caption{An instance of temporal behavior of the signals.}
	\label{fig:temp}
	\vspace{-3mm}
\end{figure}
\begin{figure}[t]
	\centering
	{\includegraphics[width =3.3in]{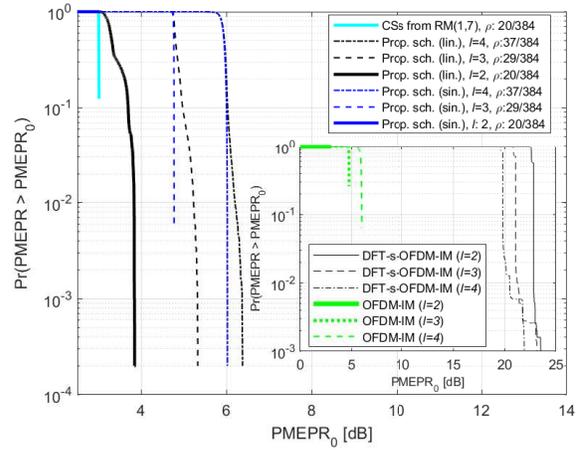}
	} 
	\vspace{-2mm}
	\caption{PMEPR distribution.}
	\label{fig:pmepr}
	\vspace{-3mm}
\end{figure}
\begin{figure}[t]
	\centering
	{\includegraphics[width =3.3in]{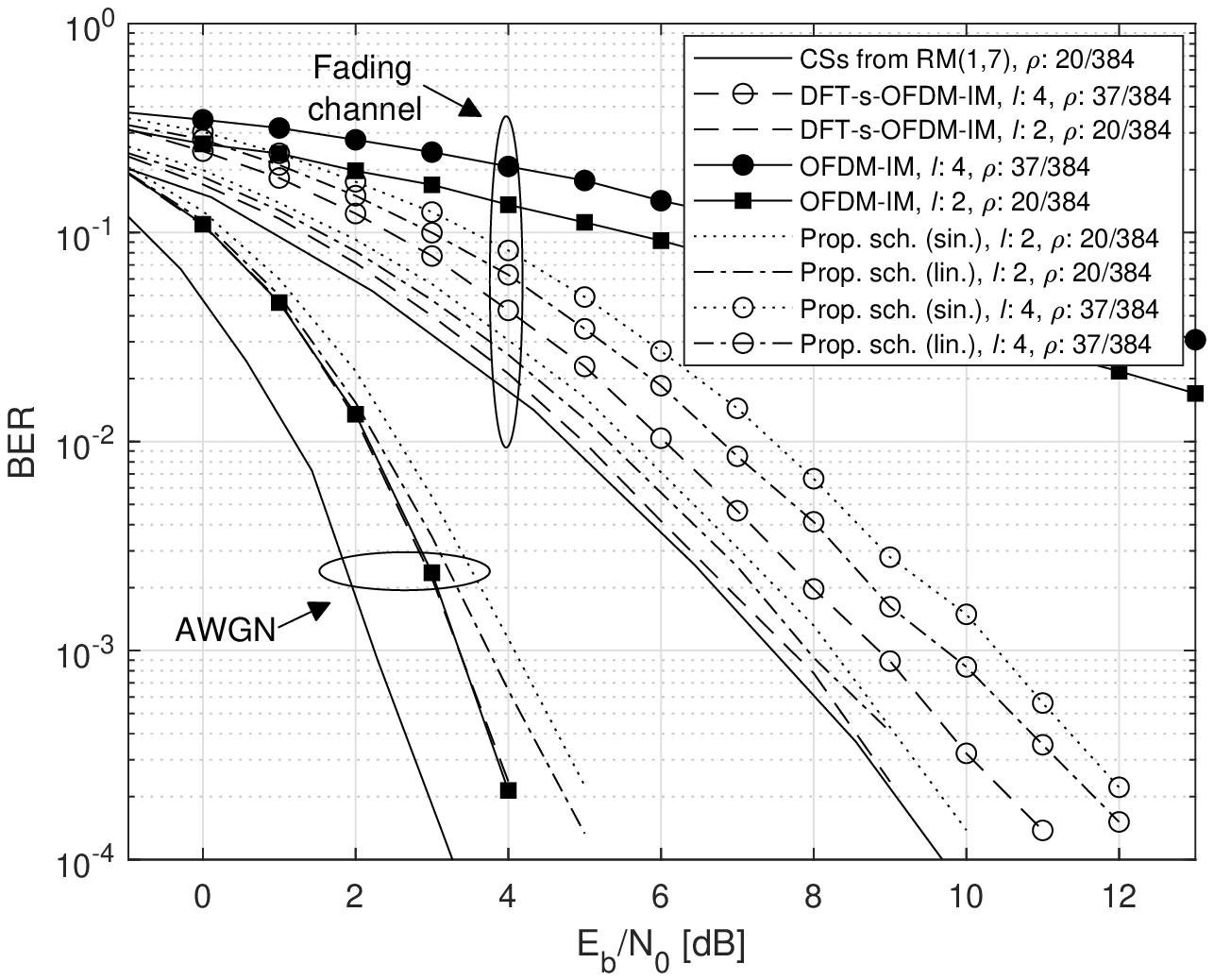}
	} 
	\vspace{-2mm}
	\caption{\ac{BER} versus $\EbNO$  in \ac{AWGN} and fading channels.}
	\label{fig:ber}
	\vspace{-3mm}
\end{figure}
\begin{figure}[t]
	\centering
	{\includegraphics[width =3.3in]{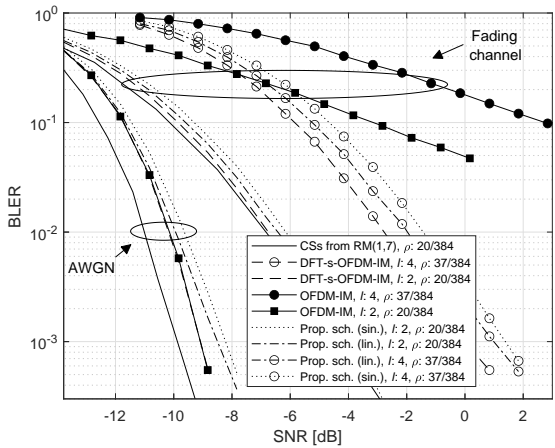}
	} 
	\vspace{-2mm}
	\caption{\ac{BLER} versus \ac{SNR} in \ac{AWGN} and fading channels.}
	\label{fig:bler}
	\vspace{-3mm}
\end{figure}

In  \figurename~\ref{fig:temp}, we plot the time-domain signals obtained with \ac{DFT-s-OFDM}-\ac{IM} and the proposed scheme for $\numberofIndices=2$. The \ac{DFT-s-OFDM}-\ac{IM} shows very high \ac{PMEPR} as compared to the one generated with the proposed scheme as the indices are represented as Dirichlet-sinc pulses. With sinusoidal chirps, the instantaneous signal power never exceeds $2$ (i.e., $\sqrt{2}$ as amplitude). This can be understood either by  the properties of a \ac{CS} or the summation of two circularly-shifted chirps. Similar behavior is also observed for linear chirps.

In \figurename~\ref{fig:pmepr}, the \ac{PMEPR} distributions are provided for different schemes. The signal is over-sampled to measure \ac{PMEPR} accurately. The \ac{PMEPR} is always less than or equal to $3$ dB for \acp{CS} from \ac{RM} codes and sinusoidal chirps for $\numberofIndices=2$. However, the distortion on linear chirps due to the truncation is higher than the one for sinusoidal chirps  for $\numberOfOccupiedSubcarriers=300$. Therefore, the \ac{CS} is not accurately formed with linear chirps under our simulation settings and the maximum \ac{PMEPR} is slightly higher than $3$~dB. 
For $\numberofIndices=3$ and $\numberofIndices=4$, \acp{PMEPR} are still limited for the proposed scheme while offering a higher \ac{SE} as compared to \acp{CS} from \ac{RM} codes. \ac{OFDM}-\ac{IM} results in \ac{PMEPR} distributions are similar to the ones for the proposed scheme. On the other hand, the \ac{DFT-s-OFDM}-\ac{IM} causes signals with very high \acp{PMEPR}. 

In \figurename~\ref{fig:ber} and \figurename~\ref{fig:bler}, we compare the error-rate performance in \ac{AWGN} and fading channels for $\numberofIndices=\{2,4\}$. 
Since the proposed scheme utilizes \ac{IM} and the \ac{DFT} is an orthogonal transformation, it inherits the structural properties of orthogonal \ac{FSK} with coherent detection at the receiver \cite{proakisfundamentals}.   For the proposed scheme, the receiver equalizes the signal in the frequency domain even in the \ac{AWGN} channel because of \ac{FDSS}. Therefore, a flatter response improves the \ac{BER} result in both \ac{AWGN} and fading channels as demonstrated in \figurename~\ref{fig:ber}. Since \ac{FDSS} for linear chirp is flatter than that of sinusoidal chirp \cite{sahin_2020}, the proposed scheme with linear chirps yields better a \ac{BER} performance. In \ac{AWGN}, the proposed schemes with different \ac{FDSS} operate in the range of $3.5$~dB - $4.5$~dB for $\EbNO$ at 1e-3 \ac{BER}.  The \acp{CS} from \ac{RM} code is superior to all other schemes for $\numberofIndices=2$ and provides $1.5$~dB gain. However, this gain diminishes in the fading channel and the schemes result in similar performances. The maximum difference is around $0.5$~dB at  $\EbNO=8$~dB approximately for $\numberofIndices=2$. Note that the \ac{ML} detector for \acp{CS} from \ac{RM} code runs the recursive algorithm for 16 times for different permutations while the proposed scheme only relies on a single $\numberOfShifts$-\ac{FFT}. For $\numberofIndices=4$, the proposed scheme offers flexibility to transmit a larger number of bits at the cost of $3$~dB \ac{PMEPR} increment. Increasing the number of transmitted bits from $20$ to $37$ causes approximately $4$~dB \ac{SNR} loss as shown in \figurename~\ref{fig:bler} as the received energy per bit is approximately halved. Although the \ac{DFT-s-OFDM}-\ac{IM} performs $1$~dB better than the proposed schemes for $\numberofIndices=4$, a larger power back-off is required for plain \ac{DFT-s-OFDM} (see \ac{PMEPR} distributions in \figurename~\ref{fig:pmepr}), which substantially offset this gain.  The performance of \ac{OFDM}-\ac{IM} is worse than the proposed scheme since it does not fully exploit the frequency selectivity. The slopes of the \ac{BER} and \ac{BLER} curves for the proposed scheme under  the fading channel are also noticeably higher than the ones for \ac{OFDM}-\ac{IM} as the transmitted signals are wideband.

\section{Concluding Remarks}

In this study, we prove that \acp{CS} and chirps are related to each other. We exemplify that Bessel functions and Fresnel integrals can be useful for generating \acp{GCP}. By utilizing this relationship, we then develop a low-complexity low-\ac{PMEPR} wide-band \ac{IM} with chirps. We discuss a generalization of the proposed scheme, which yields a trade-off between \ac{SE} and \ac{PMEPR}. 
We compare the proposed scheme with \ac{OFDM}-\ac{IM}, \ac{DFT-s-OFDM}-\ac{IM}, and the \acp{CS} from \ac{RM} codes. The proposed scheme offers significantly better \ac{PMEPR} performance than \ac{DFT-s-OFDM}-\ac{IM} while exploiting frequency selectivity without any other additional method. Since it does not utilize a coset term needed for the \ac{CS} from \ac{RM}, a low-complexity decoder can also be utilized at the receiver for the proposed scheme.



\bibliographystyle{IEEEtran}
\label{sec:conclusion}

\bibliography{references}

\end{document}